\documentclass[10 pt,conference]{IEEEtran}

\usepackage{amsmath,times,amssymb,epsfig,nicefrac,color}
\usepackage{amsfonts,graphics,subfigure,epsfig,graphics}
\usepackage[mathcal]{euscript} % redefine \mathcal to be \EuScript;
\usepackage{mathrsfs}
\usepackage{cite,calc}
\usepackage{euler}
\newcommand\nc\newcommand
\nc\bfa{{\boldsymbol a}}\nc\bfA{{\boldsymbol A}}\nc\cA{{\mathcal A}}
\nc\bfb{{\boldsymbol b}}\nc\bfB{{\boldsymbol B}}\nc\cB{{\mathcal B}}
\nc\bfc{{\boldsymbol c}}\nc\bfC{{\boldsymbol C}}\nc\cC{{\mathcal C}}
\nc\sC{{\mathscr C}}
\nc\bfd{{\boldsymbol d}}\nc\bfD{{\boldsymbol D}}\nc\cD{{\mathcal D}}
\nc\bfe{{\boldsymbol e}}\nc\bfE{{\boldsymbol E}}\nc\cE{{\mathcal E}}
\nc\bff{{\boldsymbol f}}\nc\bfF{{\boldsymbol F}}\nc\cF{{\mathcal F}}
\nc\bfg{{\boldsymbol g}}\nc\bfG{{\boldsymbol G}}\nc\cG{{\mathcal G}}
\nc\bfh{{\boldsymbol h}}\nc\bfH{{\boldsymbol H}}\nc\cH{{\mathcal H}}
\nc\bfi{{\boldsymbol i}}\nc\bfI{{\boldsymbol I}}\nc\cI{{\mathcal I}}
\nc\bfj{{\boldsymbol j}}\nc\bfJ{{\boldsymbol J}}\nc\cJ{{\mathcal J}}
\nc\bfk{{\boldsymbol k}}\nc\bfK{{\boldsymbol K}}\nc\cK{{\mathcal K}}
\nc\bfl{{\boldsymbol l}}\nc\bfL{{\boldsymbol L}}\nc\cL{{\mathcal L}}
\nc\bfm{{\boldsymbol m}}\nc\bfM{{\boldsymbol M}}\nc\sM{{\mathscr M}}
\nc\bfn{{\boldsymbol n}}\nc\bfN{{\boldsymbol N}}\nc\cN{{\mathcal N}}
\nc\bfo{{\boldsymbol o}}\nc\bfO{{\boldsymbol O}}\nc\cO{{\mathcal O}}
\nc\bfp{{\boldsymbol p}}\nc\bfP{{\boldsymbol P}}\nc\cP{{\mathcal P}}
\nc\bfq{{\boldsymbol q}}\nc\bfQ{{\boldsymbol Q}}\nc\cQ{{\mathcal Q}}
\nc\bfr{{\boldsymbol r}}\nc\bfR{{\boldsymbol R}}\nc\cR{{\mathcal R}}
\nc\bfs{{\boldsymbol s}}\nc\bfS{{\boldsymbol S}}\nc\cS{{\mathcal S}}
\nc\bft{{\boldsymbol t}}\nc\bfT{{\boldsymbol T}}\nc\cT{{\mathcal T}}
\nc\bfu{{\boldsymbol u}}\nc\bfU{{\boldsymbol U}}\nc\cU{{\mathcal U}}
\nc\bfv{{\boldsymbol v}}\nc\bfV{{\boldsymbol V}}\nc\cV{{\mathcal V}}
\nc\bfw{{\boldsymbol w}}\nc\bfW{{\boldsymbol W}}\nc\cW{{\mathcal W}}
\nc\bfx{{\boldsymbol x}}\nc\bfX{{\boldsymbol X}}\nc\cX{{\mathcal X}}
\nc\bfy{{\boldsymbol y}}\nc\bfY{{\boldsymbol Y}}\nc\cY{{\mathcal Y}}
\nc\bfz{{\boldsymbol z}}\nc\bfZ{{\boldsymbol Z}}\nc\cZ{{\mathcal Z}}

\DeclareMathOperator{\rank}{rank}

\newcommand{\rdss}{\mathrm{RDSS}}
\newcommand{\indx}{\mathrm{INDEX}}
\newcommand{\len}{\mathrm{len}}
\newcommand{\minrank}{\mathrm{minrank}}

\newtheorem{theorem}{Theorem}
\newtheorem{definition}{Definition}
\newtheorem{lemma}[theorem]{Lemma}
\newtheorem{proposition}[theorem]{Proposition}
\newtheorem{corollary}[theorem]{Corollary}

\newtheorem{remark}{\indent Remark}

\newcommand\ff{{\mathbb F}}

\begin{document}
\title{On a Duality Between Recoverable Distributed Storage and Index Coding}
\author{\IEEEauthorblockN{Arya Mazumdar}
\IEEEauthorblockA{Department of ECE\\
University of Minnesota--
Twin Cities \\Minneapolis, MN  55455\\ email: \texttt{arya@umn.edu}}
}
%\thanks{This work was supported in part by the some grants }

\maketitle

\allowdisplaybreaks

\begin{abstract}
In this paper, we introduce a model of a single-failure locally recoverable
distributed storage system. This model appears to give rise to a problem seemingly
dual of the well-studied index coding problem. The relation between 
the dimensions of an optimal index code and optimal distributed storage code
of our model has been  established in this paper. We also show some extensions
to vector codes.
\end{abstract}

\section{Introduction}

Recently,  local repair property of error-correcting codes is the center of a lot of research activity. 
In a distributed storage system, a single server failure is the most common error-event, and in that 
case, the aim is to reconstruct the content of the server from as few other servers as possible (or by downloading 
minimal amount of data from other servers). The study of such {\em regenerative}  storage systems was
initiated in \cite{dimakis2010network} and then followed up in several recent works. In \cite{gopalan2012locality},
a particularly neat characterization of a local repair property is provided. It is assumed that,  each symbol of an encoded 
message is stored at a different
node in the network (since the symbol alphabet is
  unconstrained, a symbol could represent a packet or block of bits of
  arbitrary size). 
  Accordingly,
\cite{gopalan2012locality} investigates codes allowing any
single symbol of any codeword to be recovered from at most a constant
number of other symbols of the codeword, i.e., from a number of
symbols that does not grow with the length of the code.  

The work of \cite{gopalan2012locality} is then further generalized to several directions
and a number of impossibility results regarding, as well as construction of, {\em locally
repairable codes} were presented 
(see, for example, \cite{papailiopoulos2012locally,tamo2013optimal,cadambe2013upper,silberstein2013optimal,kamath2012codes}),
culminating in very recent construction of \cite{barg2013family}.

However, the topology of the network of
distributed storage system is missing from the above definition of local repairability.
Namely, all servers are treated equally irrespective of their physical positions, proximities, 
and connections.
Here we take a step to include that into consideration.
 We study the case when the topology of the storage system is fixed and the network of storage
 is given by a graph. In our model, the servers are represented by the vertices of a graph, and
 two servers are connected by an edge if it is easier to establish up-or-down link between them, for reasons
 such as physical locations  of the servers, architecture of the distributed system or 
 homogeneity of softwares, etc. It turns out that,
our model is closely related to the following {\em index coding} problem on a side information graph.
In this paper, we formalize this relation.

\subsection{Index Coding}
A very natural ``source coding'' problem on a network, called the {\em index coding}, was introduced in
\cite{bar2006index}, and since then is a subject of extensive research.  In the index coding problem a {\em side information} graph
$G(V,E)$ is given. Each vertex $v \in V$ represents a receiver that is interested in knowing a uniform random variable $Y_v \in \ff_q$.
For any $v\in V$, define $N(v)= \{u\in V: (v,u) \in E\}$ to be 
the neighborhood of $v$.  The receiver at $v$ knows the values of the  variables $Y_u, u \in N(v)$.
How much information should a broadcaster  transmit, such that every receiver knows
the value of its desired random variable? Let us give the formal definition from \cite{bar2006index},
adapted for $q$-ary alphabet here.
\begin{definition}
An  {\em index
code} $\cC$ for $\ff_q^n$ with side information graph $G(V,E), V = \{1,2,\dots,n\},$ is a set of
codewords in $\ff_q^{\ell}$ together with:
\begin{enumerate}
\item An encoding function $f$ mapping inputs in $\ff_q^n$
to codewords, and
\item A set of deterministic decoding functions $g_1,\dots,g_n$ such
that $g_i\Big(f(Y_1,\dots,Y_n), \{Y_j: j \in N(i)\}\Big) = Y_i$ for every $i=1, \dots,n$.
\end{enumerate}
 The encoding and decoding functions 
 depend on G. The integer $\ell$ is called  the length of $\cC$, or $\len(\cC)$. 
Given a graph $G$ the minimum possible length of an index code is denoted by $\indx_q(G)$.
\end{definition} 

In \cite{bar2006index}, a connection has been made with the length of an index code to a quantity called
the minrank of the graph. Suppose, $A = (a_{ij})$ be an $n \times n$ matrix over $\ff_q$. It is said that $A$ {\em fits}
$G(V,E)$ over $\ff_q$ if $a_{ii} \ne 0$ for all $i$ and $a_{ij} =0$ whenever $(i,j) \notin E$ and $i \ne j$. 
\begin{definition}
The minrank of a graph $G(V,E)$ over $\ff_q$ is defined to be,
\begin{equation}
\minrank_q(G) = \min\{\rank_{\ff_q}(A): A \text{ fits } G\}.
\end{equation}
\end{definition}   
It was shown in \cite{bar2006index}, that,
\begin{equation}
\indx_q(G) \le \minrank_q(G),
\label{eq:bary}
\end{equation} 
and indeed, $\minrank_q(G)$ is the minimum length of an index code on $G$ when the encoding function, and the decoding functions
are all {\em linear}. The above inequality can be strict in many cases \cite{alon2008broadcasting,lubetzky2009nonlinear}.

In \cite{alon2008broadcasting}, the problem of index coding is further generalized. We only describe here what is 
important for our context. Just for this part, assume $q=2$. 
To characterize the optimal size of an index code, \cite{alon2008broadcasting} introduces
the notion of a {\em confusion graph}. Two input strings, $\bfx= (x_1, \dots, x_n), \bfy = (y_1, \dots, y_n) \in \ff_2^n$
 are called {\em confusable} if there exists some $i \in \{1, \dots, n\}$, such that $x_i \ne y_i$, but 
 $x_j = y_j,$ for all $j \in N(i)$. In the confusion graph of $G$, total number of vertices are $2^n$, and each vertex
 represents a different $\{0,1\}$-string of length $n$. There exists an edge between two vertices
 if and only if the corresponding two strings are confusable with respect to the graph $G$.
 The maximum size of an {\em independent set} of the confusion graph is denoted by $\gamma(G)$.
 
 However, the confusion graph and $\gamma(G)$ in \cite{alon2008broadcasting} were used
 as  tools to characterize the the {\em rate} of index coding; they were not used to model any immediate
 practical problem. In this paper, we show that, this notion of {\em confusable} strings fits perfectly
 to the situation of {\em local recovery} of a distributed storage system. Namely, $\gamma(G)$, in our problem
 becomes the largest possible size of a locally recoverable code for a system with topology given by $G$.

\subsection{Organization}

The paper is organized in the following way. In Section \ref{sec:rdss}, we introduce formally
the model of a recoverable distributed storage system. % inspired by \cite{alon2008broadcasting}. 
The notion
of an optimal recoverable distributed storage code given a graph and its relation to the optimal 
index code is also described here. In Section \ref{sec:dual}, we provide an algorithmic
proof of the main duality relation of the index code and distributed storage code. Our proof is
based on a covering argument of the Hamming space, and rely on the fact that for any given subset of the Hamming space
there exists a translation of the set, 
  that has very small overlap with the original subset. We conclude with an extension of the duality theorem 
  to vector codes and a remark on the optimal linearly recoverable distributed storage codes\footnote{After the first version of this paper appeared in arxiv,
  we were made aware of a parallel independent work \cite{shanmugam2014} where 
for {\em  vector linear codes} the duality between RDSS and index codes (see the discussion preceding  Eq.~\eqref{eq:minrk}) is proved. The authors of  \cite{shanmugam2014} use that observation
 to  give an upper bound on the optimal
linear sum rate of the multiple unicast network coding problem. In this paper we have a different focus: we show a proof of (approximate) duality for general (nonlinear) codes.}. 

\section{Recoverable distributed storage systems}\label{sec:rdss}
Consider the network of distributed storage, for example, one  of Fig.~\ref{fig:example}. As mentioned in the introduction,
the property of
two servers  connected by an edge is based on the ease of establishing a link between the servers.
It is also possible (and sensible, perhaps) to model this as a directed graph (especially when uplink and downlink 
constructions have varying difficulties). In the following, we assume that the graph is directed, and an undirected graph is just a 
special case.

%However, for the sake of clear  presentation, here we restrict ourselves only
%to the case of simple undirected graphs.

If the data of any one server is lost, we want to recover it from the {\em nearby} servers, i.e., the
ones with which it is easy to establish a link. This notion is formalized below.

Suppose, the directed graph $G(V,E)$ represents the network of storage. Each element of $V$
represents a server, and in the case of a server
failure (say, $v \in V$ is the failed server) one must be able to reconstruct its content from its
neighborhood $N(v)$. 

Given, this constraint  what is the maximum amount of information one can store in the system?
Without loss of generality, assume $V= \{1,2,\dots,n\}$ and the  variables
 $X_1, X_2,\dots,X_n$ respectively denote the content of the vertices, where, $X_i \in \ff_q, i =1,\dots,n.$
 \begin{definition}
A  {\em recoverable distributed storage system (RDSS)
code} $\cC \subseteq \ff_q^n$ with storage recovery graph $G(V,E), V = \{1,2,\dots,n\},$ is a set of
vectors in $\ff_q^{n}$ together with:
\begin{enumerate}
%\item An encoding function $g$ mapping inputs in $\ff_q^k$
%to codewords, and
\item[-] A set of deterministic recovery  functions, $f_i:\ff_q^{|N(i)|}\to \ff_q$ for $i = 1,\dots,n$ such that
for any codeword $(X_1, X_2,\dots,X_n) \in \ff_q^n,$
\begin{equation}
X_i = f_i(\{X_j: j \in N(i)\}), \quad i = 1,\dots,n.
\end{equation}
\end{enumerate}
 Again, the  decoding functions
 depend on G. The log-size of the code, $\log_q |\cC|$, is called  the dimension  of $\cC$, or $\dim(\cC)$. 
Given a graph $G$ the maximum possible dimension of an RDSS code is denoted by $\rdss_q(G)$.
\end{definition}

%The amount of information stored in the distributed storage system is  
%$H(X_1,\dots,X_n)$, the entropy
%of $X_1, X_2,\dots,X_n$ measured in $q$-ary. 
%There must exist GIven the graph $G(V,E)$, one needs to choose the functions $f_i, i = 1,\dots,n$, such that
%$H(X_1,\dots,X_n)$ is maximized.
%
%Formally, define,
%\begin{align}
%\rdss_q(G) &= \max_{f_i, i = 1,\dots,n}  H(X_1,\dots,X_n)\\
%\text{ such that } \quad X_i &= f_i(\{X_j: j \in N(i)\}), \quad i = 1,\dots,n.
%\end{align}

\begin{figure}[t]
\begin{center}
\includegraphics[width=0.5\textwidth]{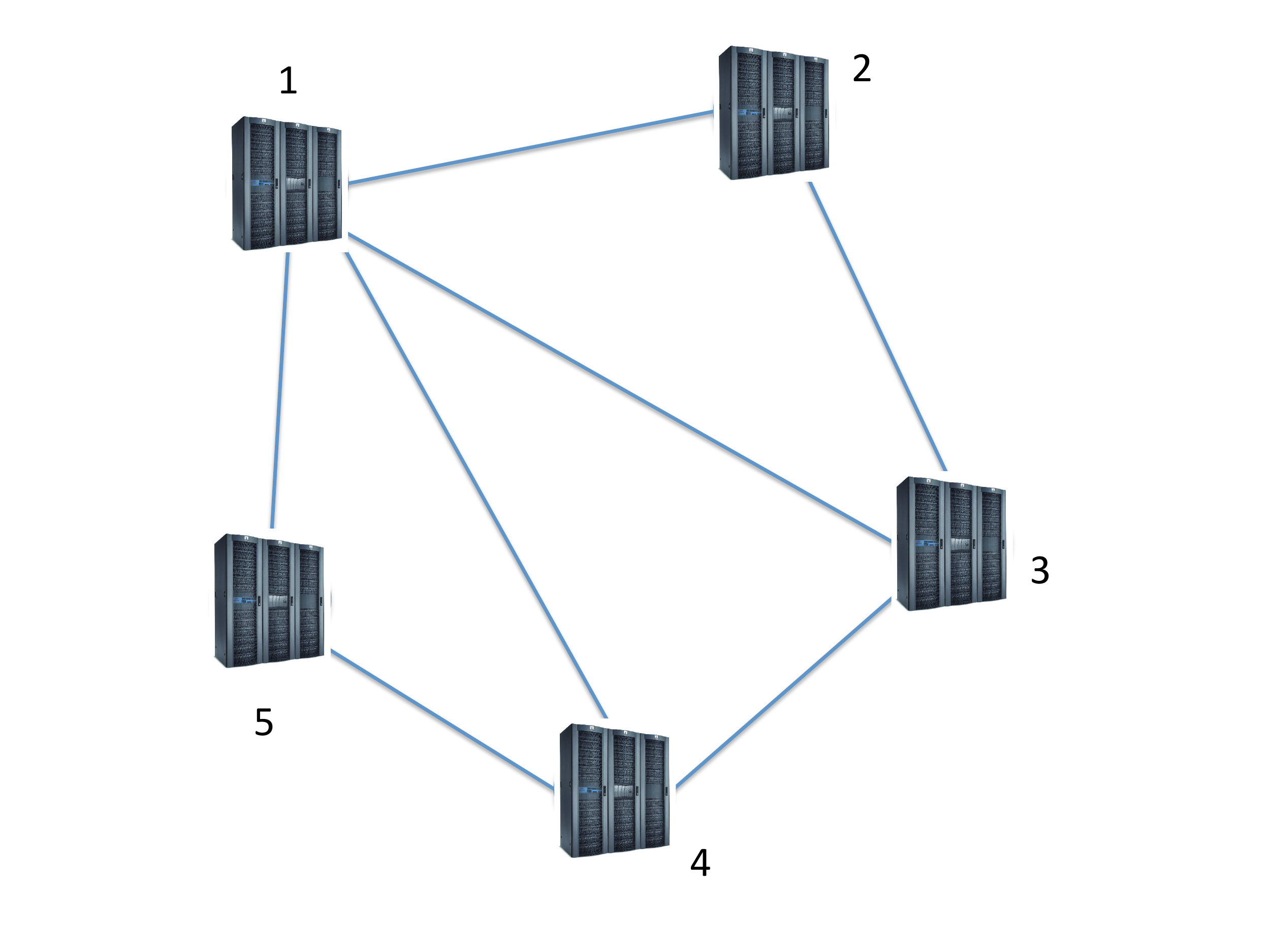}
\end{center}
\caption{Example of a distributed storage graph} 
\label{fig:example}
\end{figure}

For example, consider the graph of Fig.~\ref{fig:example} again. Here, $V = \{1,2,3,4,5\}$. The recovery sets
of each vertex (or storage nodes) are given by:
\begin{align*}
N(1) = \{2,3,4,5\}, \, N(2)  = \{1,3\}, N(3)  = \{1,2,4\}, \\
 N(4)  = \{1,3,5\},  N(5)  = \{1,4\}.
\end{align*}
Suppose, the contents of the nodes $1,2,\dots,5$ are $X_1, X_2,\dots,X_5$ respectively, where, $X_i \in \ff_q, i =1,\dots,5.$
Moreover, $X_1 = f_1(X_2,X_3,X_4,X_5), X_2 = f_2(X_1,X_3), X_3 = f_3(X_1,X_2,X_4), X_4 = f_4(X_1,X_3, X_5), X_5 = f_5(X_1, X_4).$

Assume, the functions $f_i, i =1, \dots, 5$, in this example are linear. That is, for $\alpha_{ij} \in \ff_q, 1\le i,j\le 5,$
\begin{align*}
X_1 &= \alpha_{12} X_2 +\alpha_{13} X_3 + \alpha_{14}X_4 + \alpha_{15}X_5\\
X_2 &= \alpha_{21}X_1+\alpha_{23}X_3\\
X_3 &= \alpha_{31}X_1 +\alpha_{32}X_2+\alpha_{34}X_4\\
X_4 &= \alpha_{41}X_1 + \alpha_{43}X_3 + \alpha_{45}X_5\\
X_5 &= \alpha_{51}X_1+\alpha_{54}X_4.
\end{align*}
This implies, $(X_1,X_2,\dots, X_5)$ must belong to the null-space (over $\ff_q$) of 
\[
%D(\alpha_{12}, \dots, \alpha_{54}) 
D \equiv
 \left( \begin{array}{ccccc}
1 & -\alpha_{12} & -\alpha_{13}  & - \alpha_{14} & -\alpha_{15} \\
-\alpha_{21} & 1 &  -\alpha_{23} & 0 & 0 \\
-\alpha_{31} & -\alpha_{32} & 1 & -\alpha_{34} & 0\\
-\alpha_{41} & 0 &-\alpha_{43}  & 1 & -\alpha_{45}\\
-\alpha_{51} &0 & 0 & -\alpha_{54} & 1
\end{array} \right).\] 
The dimension of the null-space of $D$ is $n$ minus the rank of $D$. Hence, it is evident
that the dimension of the RDSS code is $n -\minrank_q(G)$. 
Also, the null-space of a linear index code for $G$ is a linear RDSS
code for the same graph $G$ (see, Eq.~\eqref{eq:bary}).
From the above discussion, we have, 
\begin{equation}
\rdss_q(G) \ge n - \minrank_q(G),
\label{eq:minrk}
\end{equation}
and, $n -\minrank_q(G)$ is the maximum possible dimension of an RDSS code when
the recovery functions are all linear. 
At this point, it is tempting to make the assertion
$\rdss_q(G) = n - \indx_q(G)$, however, that would be wrong. This is shown 
in the following example.

This example is present in \cite{alon2008broadcasting}, and the distributed storage graph, a {\em pentagon}, is 
shown in Fig.~\ref{fig:pent}. For this graph, a maximum-sized binary RDSS code consists of the codewords
$\{00000,01100,00011,11011,11101\}$. The recovery functions are given by,
\begin{align*}
X_1 = X_2 \wedge X_5, 
X_2 = X_1 \vee X_3,
X_3 = X_2 \wedge \bar{X}_4, \\
X_4 = \bar{X}_3 \wedge X_5
X_5 = X_1\vee X_4.
\end{align*}
If all the recovery functions are linear, we could not have an RDSS code with 
so many codewords. Here $\rdss_2(G) = \log_2 5$. On the other hand, the minimum length of
an index code for this graph is $3$, i.e., $\indx_2(G)=3$, and this is achieved
by the following linear mappings. The broadcaster transmit $Y_1= X_2+X_3, Y_2= X_4+X_5$ and $Y_3= X_1+X_2+X_3+X_4+X_5.$
The decoding functions are, $X_1 = Y_1 + Y_2+Y_3; X_2 = Y_1+X_3; X_3 = Y_1+X_2; X_4 = Y_2+X_5; X_5 = Y_2+X_4.$

Although in general $\rdss_q(G) \ne n - \indx_q(G)$, these two quantities are
not too far from each other. In particular, for large enough alphabet, the left and right hand
sides can be arbitrarily close. This is reflected in Thm.~\ref{thm:main} below.

It is to be noted that, we refrain from using ceiling and floor functions for clarity in this paper.
In many cases, it is clear that the number of interest is not an integer and should be rounded
off to the nearest larger or smaller integer. The main results do not change 
for this.  

\subsection{Implication of the results of \cite{alon2008broadcasting}}
The result of \cite{alon2008broadcasting} can be cast in our context in the 
following way.
\begin{theorem}\label{thm:main}
Given a graph $G(V,E)$, we must have,
\begin{align}
n - &\rdss_q(G) \le \indx_q(G) \le n -\rdss_q(G)\nonumber \\
&  + \log_q\Big(\min\{n\ln q, 1+ \rdss_q(G)\ln q \}\Big). 
\end{align}
\end{theorem}
This result is purely graph-theoretic, the way it was presented
in \cite{alon2008broadcasting}. In particular, the size of maximum independent set
of the confusion graph, $\gamma(G)$ 
was identified as the size of the RDSS code, and its relation to 
the {\em chromatic number} of the confusion graph, which represents the size of the
index code was found.
 Namely the proof was dependent on the following two crucial
steps.
\begin{enumerate}
\item The {\em chromatic number} of the graph can only be so much away from
the {\em fractional chromatic number} (see, \cite{alon2008broadcasting} for detailed definition).
\item The confusion graph is {\em vertex transitive}. This implies that the maximum size of an 
independent set is equal to the number of vertices divided by the fractional chromatic number. 
\end{enumerate}
A proof of the first fact above can be found in \cite{lovasz1975ratio}.
In what follows, we give a simple {\em coding theoretic} proof of this main theorem, without using 
the notion of the confusion graph or its vertex transitivity, for completeness.

\section{The proof of the duality}\label{sec:dual}
%We will provide a coding theoretic proof of the following theorem.
%To prove Theorem \ref{thm:main}, we use  a lemma of
%coding theory, that was used to prove the famous Bassalygo-Elias bound on 
%the size of error-correcting codes. This helps us completely bypass
%the notion of confusion graph. 
We prove Theorem \ref{thm:main} with the help of following two lemmas. 
The first of them is immediate.
\begin{lemma}\label{lem:equiv}
If there exists an index code $\cC$ of length $\ell$ for a side information graph $G$ on $n$ vertices, then there exists an RDSS code
of dimension $n-\ell$ for the distributed storage graph $G$.
\end{lemma}
\begin{proof}
Suppose, the encoding and decoding functions of the index code $\cC$ are $f: \ff_q^n \to \ff_q^{\ell}$ and $g_i:\ff_q^{\ell+N(i)}\to \ff_q, i =1, \dots,n$.
There must exists some $\bfx \in \ff_q^\ell$ such that $|\{\bfy \in \ff_q^n: f(\bfy) = \bfx\}| \ge q^{n-\ell}$.
Let, $\cD_\bfx \equiv \{\bfy \in \ff_q^n: f(\bfy) = \bfx\}$ be the distributed storage code with
recovery functions,
$$
f_i(\{X_j, j \in N(i)\}) \equiv g_i(\bfx, \{X_j, j \in N(i)\}).
$$ 
\end{proof}
The second lemma is the more interesting one.
\begin{lemma}\label{lem:rdss}
If there exists an RDSS code $\cC$ of dimension $k$ for a  distributed storage graph $G$ on $n$ vertices, then there exists an index code
of length $n-k +\log_q \min\{n\ln q, 1+ k\ln q \}$ for the side information graph $G$.
\end{lemma}

To prove this result, we need the help of a number of other lemmas. 
First of all notice that, translation of any RDSS code is an RDSS code.
\begin{lemma}\label{lem:trans}
Suppose, $\cC\subseteq \ff_q^n$ is an RDSS code. Then any known translation of $\cC$ is also an 
RDSS code of same dimension. That is, for any $\bfa \in \ff_q^n$, $\cC+\bfa \equiv \{\bfy+\bfa: \bfy \in \cC\}$
is an RDSS code of dimension $\log_q |\cC|$.
\end{lemma}
\begin{proof}
Let, $(X_1, \dots, X_n) \in \cC.$ Also assume, $\bfa = (a_1, \dots, a_n)$, and $X'_i = X_i +a_i$.
 We know that, there
exist recovery functions such that,
$X_i = f_i(\{X_j: j \in N(i)\}).$

Now, $X'_i = X_i+a_i = f_i(\{X_j: j \in N(i)\}) + a_i \equiv f'_i(\{X'_j: j \in N(i)\}$.
\end{proof}

In particular,  Lemma \ref{lem:rdss}
crucially use the existence of a {\em covering} of the entire $\ff_q^n$, by translations of
an RDSS code. 
\begin{proof}[Proof of Lemma \ref{lem:rdss}]
We will show that there exists, $\cC_1, \dots, \cC_{m},$ $\cC_i \in \ff_q^n, i =1, \dots,m,$
all of which are RDSS codes of dimension $k$ such that
\begin{equation}\label{eq:cover}
\cup_{i =1}^{m} \cC_i  = \ff_q^n,
\end{equation}
where $m =  q^{n-k}\min\{n \ln q, 1+k\ln q\}$.
Assume, the above is true. Then, any $\bfy \in \ff_q^n$ must belong to 
at least one of the $C_i$s. Suppose,  $\bfy \equiv (Y_1, \dots, Y_n)\in \ff_q^n$ and $\bfy \in C_i$. Then,
the encoding function of the desired index code $\cD$ is simply given by, $f(\bfy) = i$. 
 If the recovery functions of $\cC_i$ are $f^i_j, j=1, \dots, n$, then,   the decoding functions
of $\cD$ are given by:
$$
g_j(i, \{Y_l: l \in N(j)\}) =  f^i_j(\{Y_l: l \in N(j)\}).
$$
Clearly the length of the index code
is $\log_q m = n -k + \log_q (\min\{n\ln q,1+k\ln q\})$.

\begin{figure}[t]
\begin{center}
\includegraphics[width=0.5\textwidth]{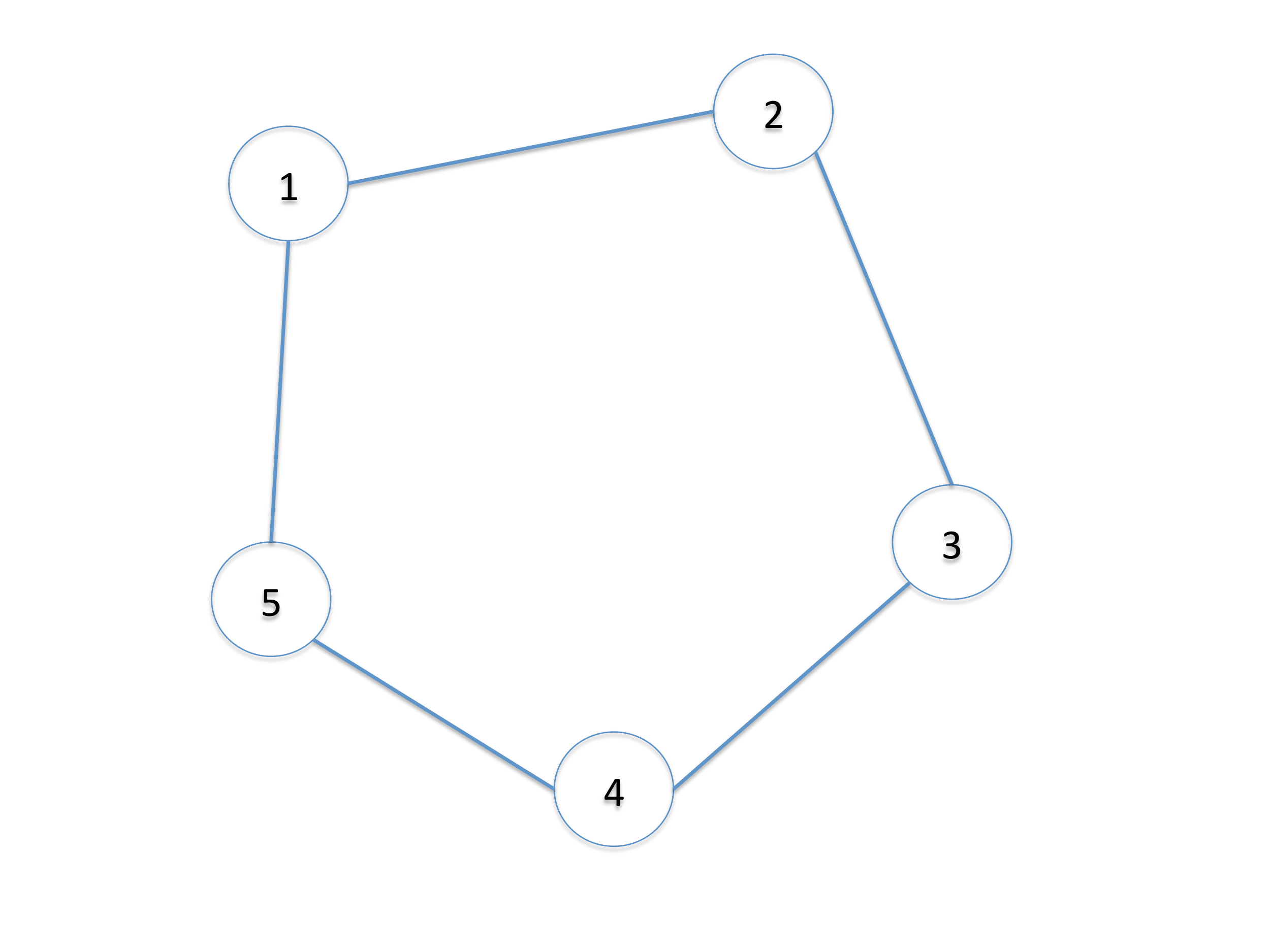}
\end{center}
\caption{A distributed storage graph (the pentagon) that shows $\rdss(G) \ne n -\indx(G)$.} 
\label{fig:pent}
\end{figure}

It remains to show the existence of RDSS codes $\cC_1, \dots, \cC_{m}$ of dimension $k$ each
with property \eqref{eq:cover}. We will show that, there exists $m$ vectors $\bfx_j, j =1, \dots, m$ such that
\begin{equation}
\cC_i = \cC + \bfx_i \equiv \{\bfy+\bfx_i : \bfy \in \cC\}.
\end{equation}
From Lemma \ref{lem:trans}, $\cC_i, i =1, \dots, m$ are all RDSS codes of dimension $k$.
Suppose, $\bfx_i, i =1, \dots, m$ are randomly and independently chosen from $\ff_q^n$.
Now,
$$
\Pr(\cup_{i =1}^{m'} \cC_i  \ne \ff_q^n) \le q^n (1-|\cC|/q^n)^{m'} < q^n e^{-m'|\cC|/q^n} \le 1,
$$
when we set $m' = q^{n-k} n \ln q \le m$ in the above expression (see \cite[Prop.~3.12]{babai1995automorphism}).

If, instead we set $m' = q^{n-k} k \ln q$ then, $\Pr(\cup_{i =1}^{m'} \cC_i  \ne \ff_q^n) \le q^{n-k}$,
which is also the expected number of points, that do not belong to any of the $m'$ translations.
To cover all these remaining points we  need at most $q^{n-k}$ other transmission. Hence,
there must exists a covering such that $q^{n-k} k \ln q + q^{n-k} = q^{n-k}(k\ln q +1)\le m$ 
translations suffice.
\end{proof}

The proof of Lemma \ref{lem:rdss} can also be given via a greedy algorithm. In the greedy algorithm 
about $\log m$ vectors are recursively chosen instead of $m$ random vectors. We provide the construction/proof next.

\subsection{A greedy algorithm for the proof of Lemma \ref{lem:rdss}}
Note that, to proof Lemma \ref{lem:rdss} we need to show
 the existence of a {\em covering} of the entire $\ff_q^n$, by translations of
an RDSS code. What we show here is that the translations themselves form a
linear subspace.
The greedy covering argument that we employ below was used to show the existence of good linear covering codes
in \cite{delsarte1986most} (see, also, \cite{goblick1963coding,cohen1983nonconstructive}), and was reintroduced 
in \cite{mazumdar2010linear} to show the existence of {\em balancing sets}.

%Indeed, the main proposition that we prove here is the following.
%\begin{proposition}
%\end{proposition}

\begin{lemma}[Bassalygo-Elias]\label{lem:elias}
Suppose, $\cC,\cB \subseteq  \ff_q^n.$ Then,
\begin{equation}
\sum_{\bfx \in \ff_q^n} \mid(\cC+\bfx) \cap \cB\mid = |\cC| |\cB|.
\end{equation}
\end{lemma}
\begin{proof}
\begin{align*}
\sum_{\bfx \in \ff_q^n} \mid(\cC+\bfx) \cap \cB\mid & = |\{(\bfx,\bfy): \bfx \in \ff_q^n, \bfy \in \cB, \bfy \in \cC+\bfx\} | \\
& = |\{(\bfx,\bfy): \bfx \in \ff_q^n, \bfy \in \cB, \bfx \in \bfy - \cC\} |\\
& = |\{(\bfx,\bfy):  \bfy \in \cB, \bfx \in \bfy - \cC\} |\\
& = |\cB| |\bfy - \cC|
 = |\cC| |\cB|,
\end{align*}
where $\bfy - \cC \equiv \{\bfy-\bfa: \bfa \in \cC\}$.
\end{proof}
For any set $\cF \subseteq \ff_q^n$, define 
\begin{equation}\label{eq:q}
Q(\cF) \equiv 1 - \frac{|\cF|}{q^n}.
\end{equation}
In words, $Q(\cF)$ denote the proportion of $\ff_q^n$ that is not covered by $\cF$.
The following property is a result of Lemma \ref{lem:elias}.
\begin{lemma}\label{lem:q}
For every subset $\cF \subseteq \ff_q^n$, 
\begin{equation}
q^{-n}\sum_{\bfx\in \ff_q^n} Q(\cF\cup (\cF+\bfx)) = Q(\cF)^2.
\end{equation}
\end{lemma}
\begin{proof}
We have,
$$
|\cF\cup (\cF+\bfx)| = 2|\cF| - |\cF\cap (\cF+\bfx)|.
$$
Therefore,
$$
 Q(\cF\cup (\cF+\bfx)) = 1 - 2|\cF| q^{-n} +  |\cF\cap (\cF+\bfx)|q^{-n},
$$
and hence,
\begin{align*}
q^{-n}\sum_{\bfx\in \ff_q^n} Q(\cF\cup (\cF+\bfx)) &= 1 - 2|\cF| q^{-n} \\ &\quad+ q^{-2n}\sum_{\bfx\in \ff_q^n}|\cF\cap (\cF+\bfx)|\\
&= 1 - 2|\cF| q^{-n} + q^{-2n}|\cF|^2\\
& = (1 - |\cF| q^{-n})^2,
\end{align*}
where in the second line we have used Lemma \ref{lem:elias}.
\end{proof}

The implication of the above lemma is the following result.
\begin{lemma}\label{lem:cover}
For every subset 
$\cF \subseteq \ff_q^n$, there exists $m = q^n |\cF|^{-1}\min\{n \ln q, 1+\ln |\cF| \}$ vectors
$\bfx_0 =0, \bfx_1, \bfx_2, \dots, \bfx_{m-1} \in \ff_q^n$, such that
$$
\cup_{i=0}^{m-1} (\cF+\bfx_i) = \ff_q^n.
$$
\end{lemma}

\begin{proof}
From Lemma \ref{lem:q}, for every subset 
$\cF \subseteq \ff_q^n$, there exists $\bfx \in \ff_q^n$ such that
$$
Q(\cF\cup (\cF+\bfx)) \le  Q(\cF)^2.
$$
For the set $\cF\equiv \cF_0$, recursively define, for i =1, 2,\dots
$$
\cF_i = \cF_{i-1} \cup (\cF_{i-1} + \bfz_{i-1}),
$$
where $\bfz_i\in \ff_q^n$ is such that, 
$$
Q(\cF_i\cup (\cF_i+\bfz_i)) \le  Q(\cF_i)^2, \quad i =0,1,\dots
$$
Clearly,
$$
Q(\cF_t) \le Q(\cF_0)^{2^t} = \Big(1- q^{-n} |\cF|\Big)^{2^t} \le e^{- q^{-n} |\cF|2^t}.
$$
At this point we can just use the argument at the end of proof of Lemma \ref{lem:rdss}, with $2^t$ plating the role of $m'$.

%If $t$ is such that $2^t > \frac{nq^n \ln q}{|\cF|}$, then 
%$$
%Q(\cF_t) < q^{-n}.
%$$
%Now, from the definition of $Q(\cdot)$ in \eqref{eq:q}, this implies $|\cF_t|  = q^n$ or $\cF_t = \ff_q^n$. 
%Moreover, if $t$ is such that, $2^t >  \frac{q^n (1+\ln |\cF|)}{|\cF|}$, then ,
%$$
%Q(\cF_t) < \frac{1}{e|\cF|}.
%$$
%But, from  \eqref{eq:q}, this implies,
%\begin{align*}
%&1 - \frac{|\cF|}{q^n} < \frac{1}{e|\cF|},\\
%\text{ or, } \quad& |\cF|^2 - q^n |\cF| +q^n/e > 0,\\
%\text{ or, } \quad &|\cF| > \frac{q^n}{2} + \frac{q^n}{2}\sqrt{1-\frac{4}{eq^n}}\\
%& \quad > \frac{q^n}{2} + \frac{q^n}{2}\Big(1-\frac{4}{eq^n}\Big) = q^n - \frac{2}{e}.
%\end{align*}
%But then, $|\cF| = q^n$, i.e., $\cF = \ff_q^n$.
%

\begin{figure}[t]
\begin{center}
\includegraphics[width=0.5\textwidth]{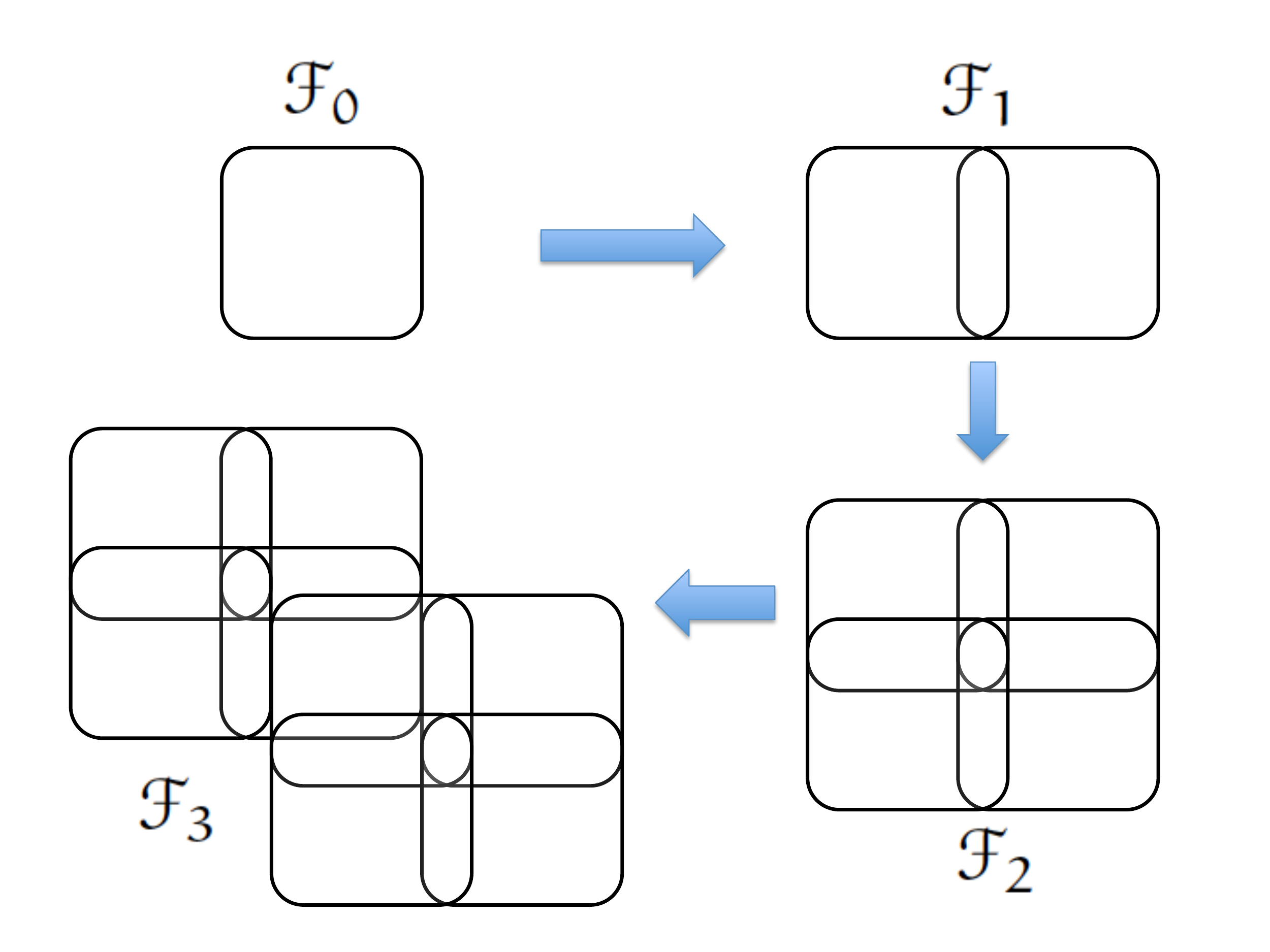}
\end{center}
\caption{The recursive construction of the sets $\cF_1, \cF_2, \cF_3$ of Lemma \ref{lem:cover}.} 
\label{fig:replicate}
\end{figure}

On the other hand $\cF_t$ contains $\cF_0$ and its $2^t -1$ translations (see, Figure \ref{fig:replicate} for 
an illustration). 
Hence,  there exists $m = \min\Big\{\frac{q^nn \ln q}{|\cF|}, \frac{q^n (1+\ln |\cF|)}{|\cF|}\Big\}$ vectors
$\bfx_0 =0, \bfx_1, \bfx_2, \dots, \bfx_{m-1} \in \ff_q^n$, such that
$$
\cup_{i=0}^{m-1} (\cF+\bfx_i) = \ff_q^n.
$$

\end{proof}
%\begin{proof}[Proof of Lemma \ref{lem:rdss}]
%We will show that there exists, $\cC_0 \equiv \cC, \cC_1, \dots, \cC_{m-1},$ $\cC_i \in \ff_q^n, i =0, \dots,m-1,$
%all of which are RDSS codes of dimension $k$ such that
%\begin{equation}\label{eq:cover}
%\cup_{i =0}^{m-1} \cC_i  = \ff_q^n,
%\end{equation}
%where $m =  q^{n-k}\min\{n \ln q, 1+k\ln q\}$.
%Assume, the above is true. Then, any $\bfy \in \ff_q^n$ must belong to 
%at least one of the $C_i$s. Suppose,  $\bfy \equiv (Y_1, \dots, Y_n)\in \ff_q^n$ and $\bfy \in C_i$. Then,
%the encoding function of the desired index code $\cD$ is simply given by, $f(\bfy) = i$. 
% If the recovery functions of $\cC_i$ are $f^i_j, j=1, \dots, n$, then,   the decoding functions
%of $\cD$ are given by:
%$$
%g_j(i, \{Y_l: l \in N(j)\}) =  f^i_j(\{Y_l: l \in N(j)\}).
%$$
%Clearly the length of the index code
%is $\log_q m = n -k + \log_q (\min\{n\ln q,1+k\ln q\})$.

To complete the proof of Lemma  \ref{lem:rdss}, as before,  we just
 show the existence of RDSS codes $\cC_0 \equiv \cC, \cC_1, \dots, \cC_{m-1}$ of dimension $k$ each
with property \eqref{eq:cover}. This is achieved by choosing $m-1$ vectors $\bfx_j, j =1, \dots, m-1$ such that
\begin{equation}
\cC_i = \cC + \bfx_i \equiv \{\bfy+\bfx_i : \bfy \in \cC\}.
\end{equation}
From Lemma \ref{lem:trans}, $\cC_i, i =1, \dots, m-1$ are all RDSS codes of dimension $k$. Moreover, from 
Lemma \ref{lem:cover}, we already know the existence of $\bfx_j, j =1, \dots, m-1$ such that property \eqref{eq:cover}
is satisfied. However, from Lemma  \ref{lem:cover} it is also clear that these $m$ vectors form a linear subspace and
can be generated by only $\log_q m$ vectors.

\begin{corollary}
For every subset 
$\cF \subseteq \ff_q^n$, there exists a linear subspace $\cD \in \ff_q^n$ such that $|\cD| = q^n |\cF|^{-1}n \ln q$ and
%$\bfx_0 =0, \bfx_1, \bfx_2, \dots, \bfx_{m-1} \in \ff_q^n$, such that
$$
\cup_{\bfx \in \cD} (\cF+\bfx) = \ff_q^n.
$$
\end{corollary}

The above result is helpful in the decoding process of the index code. If $\cC$ is an RDSS code and $\cD$ is the 
linear subspace such that $\cup_{\bfx \in \cD} (\cC+\bfx) = \ff_q^n$, then the decoding of 
the obtained index code can be performed from $\bfx \in \ff_q^{\log_q |\cD|}$ by first multiplying $\bfx$ with 
the generator matrix of $\cD$ and then shifting $\cC$ by it. Hence, if there is a polynomial time decoding algorithm for $\cC$
then there will be one for the index code. It would not be so for the case of random-choice, where we must maintain a look-up table
of size exponential in $n$.
%\end{proof}
%\begin{remark}\label{rem:one}
%Using our technique, it is possible to have a little improvement on Theorem \ref{thm:main}. Namely, the right hand inequality can
%be written as,
%\begin{align*}
%\indx_q(G) &\le n -\rdss_q(G)\\
%&+ \log_q\Big(\min\{n\ln q, \ln 2+ \rdss_q(G)\ln q \}\Big).
%\end{align*}
%This improvement comes from our result of Lemma \ref{lem:cover}, where it suffices to have $Q(\cF_t) < \frac1{2|\cF|}$
%instead of $Q(\cF_t) < \frac1{e|\cF|}$.
%\end{remark}, that is 

\begin{remark}
It is a perhaps not so surprising that the method of \cite{alon2008broadcasting}, that is the random choice,
 (or, in fact the method of \cite{lovasz1975ratio})
gives the  exact same result as the greedy algorithm method.
% (barring  the correction of Remark \ref{rem:one}).
% It is worth 
%asking exactly how different (or similar) these methods are -- in particular, exactly what part of our method  follows from the vertex transitivity of
%the confusion graph and what is due to the small ``integrality gap.''
 \end{remark}

\section{Extension to vector codes and the capacity of linear codes}
Literatures of distributed storage often considers {\em vector linear codes}
and the same is true for  \cite{shanmugam2014}. However in the context of general nonlinear
codes, vector codes do not bring any further technical novelty and can just be thought of as
codes over a larger alphabet. 

For vector index codes, as earlier, a {\em side information} graph
$G(V,E)$ is given. Each vertex $v \in V$ represents a receiver that is interested in knowing a uniform random vector $Y_v \in \ff_q^p$.
 The receiver at $v$ knows the values of the  variables $Y_u, u \in N(v)$.
% \begin{definition}
A vector  {\em index
code} $\cC$ for $\ff_q^{np}$ with side information graph $G(V,E), V = \{1,2,\dots,n\},$ is a set of
codewords in $\ff_q^{\ell p}$ ($\ell$ is the length of the code) together with:
\begin{enumerate}
\item An encoding function $f$ mapping inputs in $\ff_q^{np}$
to codewords, and
\item A set of deterministic decoding functions $g_1,\dots,g_n$ such
that $g_i\Big(f(Y_1,\dots,Y_n), \{Y_j: j \in N(i)\}\Big) = Y_i$ for every $i=1, \dots,n$.
\end{enumerate} 
Given a graph $G$ the minimum possible value of $\ell$ is denoted by $\indx^p_q(G)$ (also called the {\em broadcast capacity}). When the function $f$, $g_i$  are
 linear, for all $1\le i\le n$,
in all of their arguments in $\ff_q$, then the code is called {\em vector linear}. 
%\end{definition} 

Similar generalization is possible for the definition of RDSS codes.
% \begin{definition}
A  vector RDSS
code $\cC \subseteq  (\ff_q^p)^n$ with storage recovery graph $G(V,E), V = \{1,2,\dots,n\},$ is a set of
vectors in $\ff_q^{np}$ together with:
%\begin{enumerate}
%\item An encoding function $g$ mapping inputs in $\ff_q^k$
%to codewords, and
%\item[-] 
A set of deterministic recovery  functions, $f_i:\ff_q^{|N(i)|p}\to \ff_q^p$ for $i = 1,\dots,n$ such that
for any codeword $(X_1, X_2,\dots,X_n), X_i \in \ff_q^p,$
\begin{equation}
X_i = f_i(\{X_j: j \in N(i)\}), \quad i = 1,\dots,n.
\end{equation}
%\end{enumerate}
The normalized log-size of the code, $\frac{1}{p}\log_q |\cC|$, is called  the dimension  of $\cC$. 
Given a graph $G$ the maximum possible dimension of a vector RDSS code is denoted by $\rdss_q^p(G)$.
When the decoding functions $f_i$, $1\le i\le n$ are linear in all their  arguments (in $\ff_q$), the code is called 
{\em vector linear}. 
%\end{definition}

General (nonlinear) vector index or RDSS codes can also be thought as  scalar codes over the alphabet of size $q^p$.
Hence,
\begin{align*}
n - \rdss_q^p(G) \le &\indx_q^p(G) \\ &\le n - \rdss_q^p(G)  + \frac{\log_q (pn\ln q)}{p}.
\end{align*} 
 As a consequence, even for a constant $q$, if $p = \Omega(\log n)$, we have
 $
 \indx_q^p(G)
 $
 and $n -\rdss_q^p(G)$  differ at most by $1$
 for any graph $G$ -- and for larger $p$, this difference goes to zero.

Although, general vector codes do not lead to a different analysis, we next show that
vector linear codes can achieve a dimension  sufficiently close to $\rdss_q^p(G)$ for any graph $G(V,E).$
This should be put into contrast with results, such as \cite[Thm.~1.2]{blasiak2011lexicographic}, which show that 
a rather large gap must exist between vector linear and nonlinear index coding (or network coding) rates.
\begin{proposition}
There exists a polynomial time (in $n$) constructible vector linear RDSS code with dimension
at least $$\frac{\rdss_q^p(G)}{\beta\log n \cdot \log\log n}$$ for a large enough integer $p$ and a constant $\beta< 5$.
\end{proposition}
\begin{proof}
In \cite{shanmugam2014}, it was shown that the linear algebraic dual of a vector
linear index code is a vector linear RDSS code (see, Section \ref{sec:rdss} of this paper for 
scalar codes). This implies that, for a vector linear index code of length $\ell$, the dual code is
a vector linear RDSS code of dimension $n-\ell$. In \cite{chaudhry2011complementary},
a vector linear index code of length $\ell$ was constructed in polynomial time, such that
$
n - \ell \ge \frac{n -\indx_q^p(G)}{\alpha\log n \cdot \log\log n},
$ 
(this result of \cite{chaudhry2011complementary} was also used in \cite{shanmugam2014}), $\alpha$ is a constant (see, \cite{seymour1995packing}, the building-block  of \cite{chaudhry2011complementary},
for the value of the constant). The dual code of
this code must be a vector RDSS code of dimension $k = n-\ell$. From the above discussion, it is evident that,
$$
n -\indx_q^p(G) \ge \rdss_q^p(G)  - \frac{\log_q (pn\ln q)}{p}.
$$
Hence,
$$
k \ge \frac{\rdss_q^p(G)  - \frac{\log_q (pn\ln q)}{p}}{\alpha \log n \cdot \log\log n}.
$$
Hence if $p$ is large enough, then the statement of the theorem is proved.
\end{proof}

\begin{remark}
How large does $p$ needs to be for the above proposition to hold? It is clear that 
$p = \Omega(\log n)$ is enough to diminish the additive error term of $ \frac{\log_q (pn\ln q)}{p}$.
However, for the algorithm of \cite{chaudhry2011complementary} to work, $p$ needs to be as large
as the denominator of a linear programming solution (see, \cite{nutov2004packing}) that is used crucially in 
\cite{chaudhry2011complementary}.
 Hence $p$, depending on the number of cycles in the graph,  may
required to be exponential in $n$. 
\end{remark}

\vspace{0.1in}

\emph{Acknowledgement:}
The author thanks A. Agarwal, A. G. Dimakis and  K. Shanmugam for
%enlightening discussions and
 useful references. 
This work was supported in part by  a grant from University of Minnesota.

\bibliographystyle{abbrv}
\bibliography{/Users/arya/Documents/BIB/aryabib}

\begin{thebibliography}{10}

\bibitem{alon2008broadcasting}
N.~Alon, E.~Lubetzky, U.~Stav, A.~Weinstein, and A.~Hassidim.
\newblock Broadcasting with side information.
\newblock In {\em Foundations of Computer Science, 2008. FOCS'08. IEEE 49th
  Annual IEEE Symposium on}, pages 823--832. IEEE, 2008.

\bibitem{babai1995automorphism}
L.~Babai.
\newblock Automorphism groups, isomorphism, and reconstruction, chapter 27 of
  handbook of combinatorics.
\newblock {\em North-Holland--Elsevier}, pages 1447--1540, 1995.

\bibitem{bar2006index}
Z.~Bar-Yossef, Y.~Birk, T.~Jayram, and T.~Kol.
\newblock Index coding with side information.
\newblock In {\em Foundations of Computer Science, 2006. FOCS'06. 47th Annual
  IEEE Symposium on}, pages 197--206. IEEE, 2006.

\bibitem{blasiak2011lexicographic}
A.~Blasiak, R.~Kleinberg, and E.~Lubetzky.
\newblock Lexicographic products and the power of non-linear network coding.
\newblock In {\em Foundations of Computer Science (FOCS), 2011 IEEE 52nd Annual
  Symposium on}, pages 609--618. IEEE, 2011.

\bibitem{cadambe2013upper}
V.~Cadambe and A.~Mazumdar.
\newblock An upper bound on the size of locally recoverable codes.
\newblock In {\em Proc. IEEE Int.\ Symp.\ Network Coding}, June 2013.

\bibitem{chaudhry2011complementary}
M.~A.~R. Chaudhry, Z.~Asad, A.~Sprintson, and M.~Langberg.
\newblock On the complementary index coding problem.
\newblock In {\em Information Theory Proceedings (ISIT), 2011 IEEE
  International Symposium on}, pages 244--248. IEEE, 2011.

\bibitem{cohen1983nonconstructive}
G.~Cohen.
\newblock A nonconstructive upper bound on covering radius.
\newblock {\em Information Theory, IEEE Transactions on}, 29(3):352--353, 1983.

\bibitem{delsarte1986most}
P.~Delsarte and P.~Piret.
\newblock Do most binary linear codes achieve the goblick bound on the covering
  radius?(corresp.).
\newblock {\em Information Theory, IEEE Transactions on}, 32(6):826--828, 1986.

\bibitem{dimakis2010network}
A.~G. Dimakis, P.~B. Godfrey, Y.~Wu, M.~J. Wainwright, and K.~Ramchandran.
\newblock Network coding for distributed storage systems.
\newblock {\em IEEE Trans.\ Inform.\ Theory}, 56(9):4539--4551, Sep. 2010.

\bibitem{goblick1963coding}
T.~J. Goblick.
\newblock {\em Coding for a discrete information source with a distortion
  measure}.
\newblock PhD thesis, Massachusetts Institute of Technology, 1963.

\bibitem{gopalan2012locality}
P.~Gopalan, C.~Huang, H.~Simitci, and S.~Yekhanin.
\newblock On the locality of codeword symbols.
\newblock {\em IEEE Trans.\ Inform.\ Theory}, 58(11):6925--6934, Nov. 2012.

\bibitem{kamath2012codes}
G.~M. Kamath, N.~Prakash, V.~Lalitha, and P.~V. Kumar.
\newblock Codes with local regeneration.
\newblock {\em arXiv preprint arXiv:1211.1932}, 2012.

\bibitem{lovasz1975ratio}
L.~Lov{\'a}sz.
\newblock On the ratio of optimal integral and fractional covers.
\newblock {\em Discrete mathematics}, 13(4):383--390, 1975.

\bibitem{lubetzky2009nonlinear}
E.~Lubetzky and U.~Stav.
\newblock Nonlinear index coding outperforming the linear optimum.
\newblock {\em Information Theory, IEEE Transactions on}, 55(8):3544--3551,
  2009.

\bibitem{mazumdar2010linear}
A.~Mazumdar, R.~M. Roth, and P.~O. Vontobel.
\newblock On linear balancing sets.
\newblock {\em Advances in Mathematics of Communications (AMC)}, 4(3):345--361,
  2010.

\bibitem{nutov2004packing}
Z.~Nutov and R.~Yuster.
\newblock Packing directed cycles efficiently.
\newblock In {\em Mathematical Foundations of Computer Science 2004}, pages
  310--321. Springer, 2004.

\bibitem{papailiopoulos2012locally}
D.~S. Papailiopoulos and A.~G. Dimakis.
\newblock Locally repairable codes.
\newblock In {\em Proc.\ Int.\ Symp.\ Inform.\ Theory}, pages 2771--2775,
  Cambridge, MA, July 2012.

\bibitem{seymour1995packing}
P.~D. Seymour.
\newblock Packing directed circuits fractionally.
\newblock {\em Combinatorica}, 15(2):281--288, 1995.

\bibitem{shanmugam2014}
K.~Shanmugam and A.~G. Dimakis.
\newblock Connections between index coding, locally repairable codes and the
  multiple unicast problem.
\newblock {\em personal communication}, 2014.

\bibitem{silberstein2013optimal}
N.~Silberstein, A.~S. Rawat, O.~O. Koyluoglu, and S.~Vishwanath.
\newblock Optimal locally repairable codes via rank-metric codes.
\newblock {\em preprint, arXiv:1301.6331}, 2013.

\bibitem{barg2013family}
I.~Tamo and A.~Barg.
\newblock A family of optimal locally recoverable codes.
\newblock {\em arXiv preprint arXiv:1311.3284}, 2013.

\bibitem{tamo2013optimal}
I.~Tamo, D.~S. Papailiopoulos, and A.~G. Dimakis.
\newblock Optimal locally repairable codes and connections to matroid theory.
\newblock preprint, arXiv:1301.7693, 2013.

\end{thebibliography}
\end{document}